\documentclass{autart}
\usepackage{graphicx}      
\usepackage{natbib}        
\usepackage{amsmath}
\usepackage{mathptmx} 
\usepackage{times} 
\usepackage{amssymb}  
\usepackage{amsfonts}
\usepackage{mathtools}
\usepackage{color}
\usepackage{grffile}
\usepackage{subfigure}
\usepackage{enumerate}
\usepackage{times}
\newtheorem{remark}{Remark}
\newtheorem{theorem}{Theorem}
\newtheorem{lemma}{Lemma}

\newenvironment{proof}{\par{\em{Proof: }}}{\hfill$\square$\\}

\newcommand{\matr}[1]{\begin{pmatrix} #1 \end{pmatrix}}
\renewcommand{\thefootnote}{\alph{footnote}}
\begin{document}
\begin{frontmatter}

\title{Angle-based formation stabilization and maneuvers in port-Hamiltonian form with bearing and velocity measurements}

\author[First] {Ningbo Li,} 
\author[Second] { Pablo Borja,}
\author[First] { Arjan van der Schaft,}
\author[First] { Jacquelien M. A. Scherpen} 
\address[First] {Jan C. Willems Center for Systems and Control, University of Groningen, The Netherlands (e-mail: ningbo.li@rug.nl, a.j.van.der.schaft@rug.nl, and j.m.a.scherpen@rug.nl).}
\address[Second]{School of Engineering, Computing and Mathematics, University of Plymouth, Plymouth, UK. (e-mail: pablo.borjarosales@plymouth.ac.uk)}

\begin{abstract}                
 This paper proposes a port-Hamiltonian framework for angle-based formation stabilization and maneuvers using bearing and velocity measurements with an underlying triangulated Laman graph. The corresponding port-Hamiltonian controller is designed using virtual couplings on the errors of angle constraints in angle space and then the angle constraints and agent actuators are mapped by the constraint Jacobian, which can be applied to other formation constraints. In addition, due to the fact that the port-Hamiltonian model allows for complex and heterogeneous agent dynamics, our framework can be extended to networks with different agent dynamics and formation constraints. To avoid unavailable distance terms in the control law, an estimator is designed based on port-Hamiltonian theory and the property that energy is coordinate-free for different sensor modalities using bearing and velocity measurements, which permits our framework to inject damping for the formation maneuvers. Furthermore, several maneuvers are analyzed under both considerations of stabilization and transient performance. Simulations are performed to illustrate the effectiveness of the approach. 
\end{abstract}

\end{frontmatter}
\section{Introduction}

Over the last two decades, formation control has attracted extensive interest due to its applications in many domains. See the survey papers \cite{beard2001coordination}, \cite{ren2007information}, \cite{oh2015survey}, \cite{cortes2017coordinated}, \cite{liu2018survey}, and \cite{chen2019control} and the references therein. According to the variables that define the geometric shapes, the formation control can be categorized as displacement-, distance-, bearing-, and angle-based formations. While satisfying the different geometric constraints, the formation can achieve various maneuvers, e.g., distances are invariant to translations and rotations, and bearings are invariant to translations and scales. Compared with the above two kinds of constraints, angles have the advantage of being invariant to all three maneuvers, i.e., translations, rotations, and scales. 

From the control point of view, the sensors that agents are equipped with, i.e., the information that agents have access to, play a key role in the design of the corresponding controllers. Distance, bearing, and angle represent partial information about positions. The corresponding control laws to achieve these constraints can also be proposed using partial information on positions. Therefore, agents may require fewer onboard sensors, which reduces the hardware cost and introduces fewer measurement errors. There has been abundant literature on this topic in recent years, such as \cite{nuno2020distributed} and \cite{ji2007distributed} for displacement formations; \cite{anderson2008rigid} and \cite{cao2011formation} for distance formations; \cite{zhao2019bearing} and \cite{trinh2018bearing} for bearing formations; \cite{jing2019angle}, \cite{chen2020angle}, and \cite{li2021angle} for angle formations. 

In this paper, we study the case where the sensing capability of the agents is based on bearing and velocity measurements, and their interaction topology is constrained by angles. We remark that angle constraints require less information on the agents compared with displacement-, distance-, and bearing-based formations. In particular, it is invariant to a large class of maneuvers, such as translation, rotation, and scaling, which means the group of agents can achieve these corresponding maneuvers while satisfying angle-based constraints.  

Recently, passivity-based control (PBC) techniques in combination with the port-Hamiltonian (pH) modeling framework have found favor for the design of formation controllers, such as in \cite{vos2014formation}, \cite{stacey2015passivity}, \cite{xu2018formation}, \cite{li2022passivity}. The advantages of this approach can be summarized as follows: on the one hand, it allows for complex and heterogeneous agent dynamics. Different from most of the existing literature, where the agents are modeled as homogeneous dynamics, and the formation geometric constraints are limited to one kind of position, distance, or bearing, the pH approach can be further applied to heterogeneous systems where the agents are modeled as nonlinear dynamics with the formation shape defined by different kinds of constraints. On the other hand, the passivity property enables the flexibility and scalability of the network due to the fact that the composition of the passive systems is still passive (see \cite{van2000l2}, \cite{khalil2002nonlinear}). In the meantime, according to the theory of pH systems on graphs (see \cite{van2013port}), the Dirac structure underlying the pH dynamics incorporates the network structure, enabling the scalability of the network. Moreover, passivity-based decentralized controllers allow the agents to exert forces based on different relative information with respect to their neighbors, such as relative position, distance, and bearing. 

Customarily, in PBC, the control objectives are achieved by virtual couplings where the virtual springs determine the formation by shaping the energy function of the network, while the virtual dampers modify the transient response by injecting damping. However, the resulting control law in passivity-based approaches usually requires the agents to have complete information of relative position, even if the interaction topology of agents is represented only by angles. To solve this problem, we propose a passive estimator based on the image depth compensator in \cite{mahony2012port}, and its later application for bearing formation control in \cite{stacey2015passivity}, to estimate the unavailable distance information by using velocity and bearing measurements. As a result, the control law only contains bearing and velocity terms.  


Some literature has also studied angle-based formation recently. In \cite{basiri2010distributed} and \cite{chen2020angle}, an intuitive control law is proposed for using only local bearing information and proves the stability by linearization. In \cite{chen2022globally}, an angle-induced linear constraint of each triangle in an angle rigid formation is used to propose the formation control law with local relative position measurements. However, a suitable Lyapunov function for stability analysis is not given in these three papers. In \cite{jing2019angle}, a gradient-like control law is proposed for a single integrator model using bearing and distance information. In addition, the technique of mismatched measurements is used in \cite{chen2021maneuvering} to investigate the maneuvering and stabilization of formation control, and the resulting control law also contains bearing and distance information. In our proposed approach, the agents achieve angle-based formation using bearing and velocity measurements---rather than distance as in the aforementioned literature---by proposing an estimator. Hence, The velocity information permits injecting damping into the system, improving the transient performance, especially for time-varying maneuvers and systems with complex dynamics. In addition, the aforementioned literature on angle-based formation all considers the dynamics of single integrators, while we consider the dynamics of double integrators with the damping term, which is more practical from a physical point of view, but more complicated in terms of controller design and stability analysis. 

Furthermore, as we discussed before, angle constraints permit more maneuvers than distance and bearing constraints. We further investigated some maneuvers of the angle-based formation. Regarding the scale and orientation control, we choose two agents as leaders and control the scale and orientation of the formation by specifying a desired relative position between these two agents. Concerning velocity tracking, we assume all agents know the desired velocity.  


The main contributions of our approach can be summarized in two aspects:
\begin{itemize}
 \item [\textbf{(i)}]We propose a framework for angle-based formation, considering double integrators with damping terms, which can be adapted to more practical physical systems. Due to the pH modeling framework, in addition to angle-based formations, this approach is also applicable to other types of formations and dynamics. Therefore, it can be used to establish a general framework for multi-agent formation with heterogeneous constraints and dynamics. To the best of our knowledge, existing research on angle-based formation only considers the dynamics of single integrators.
 \item [\textbf{(ii)}] Using pH theory and the fact that energy is coordinate-free, we propose a passive estimator for the distance using bearing and velocity measurements. Compared with the literature using bearing and distance measurements for angle-based formation, our framework can inject damping into the system by using the velocity term, improving transient performance for the following design of formation maneuvers. 
\end{itemize}

The rest of the paper is structured as follows. The preliminaries and problem formulation are introduced in Section 2. The control architecture and the stability analysis are developed in Sections 3 and 4. The formation maneuvers are investigated in Section 5. The proposed approach is validated via simulations in Section 6. Then, some concluding remarks are provided in Section 7.

\section{Preliminaries and Problem Formulation}

\subsection{Preliminaries}  
To introduce the basic concepts, we consider the triangular formation as in Fig. 1. Hence, for the link $k$ between the agent $1$ and the agent $2$, we have
\begin{equation}	
z_k=q_1-q_2,
\end{equation} 
where $q_1,q_2 \in \mathbb{R}^2$ denote  the position of agents $1$ and $2$, respectively, and $z_k \in \mathbb{R}^2$ is the relative position associated with the link $k$.

\begin{figure}[!ht]
\begin{center}
\label{}
\includegraphics[width=3.5cm]{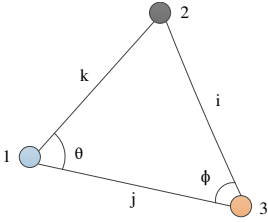}   
\caption{Triangular formation} 
\end{center}
\end{figure}

Regarding the constraints defining the formation geometric shapes, bearings and distances are simpler compared with angles. As shown in Fig. 1, bearings and distances are associated with two nodes and one edge while angles are associated with three nodes and two edges. Therefore, taking bearings and distances as an example, we give some definitions and conclusions. The distance and bearing between the agents $1$ and $2$ are given by 
\begin{equation}
\begin{split}
r_k=||z_k|| ,  \quad
s_k=\frac{z_k}{||z_k||},
\end{split}
\end{equation}
respectively. Moreover, we denote the general form of the constraints of the formation geometric shapes as $y_k\in\mathbb{Y}_k$, where $\mathbb{Y}_k$ is the constraint space. The time-evolution of $y_k$ is given by
\begin{equation}	
\dot{y}_k=L_{y_k}(z_k)\dot{z}_k,
\end{equation}
where $L_{y_k}(z_k)=\frac{\partial{y_k}}{\partial{z_k}}(z_k)$ is the constraint Jacobian. Accordingly, for the bearing $s_k$ and the distance $r_k$, the corresponding Jacobians are given, respectively, by
\begin{eqnarray}
\label{bj}	
L_{s_k}(z_k)&=&\frac{1}{r_k}(I_2-s_k{s_k}^T)\in\mathbb{R}^{2\times2},\\
L_{r_k}(z_k)&=&s_k^T\in\mathbb{R}^{1\times2},
\end{eqnarray}
where $I_2$ is the $2\times 2$ dimensional identity matrix. Henceforth, we omit the argument in the Jacobians $L_{s_k}$ and $L_{r_k}$ to simplify the notation.

Define $y_k^*$ as the desired value of the constraint $y_k$, corresponding to a desired formation. Then, we define the error
\begin{equation} \label{error}	
\tilde{y}_k:=y_k-y_k^*.
\end{equation}

Hence, the control objective is to ensure that
$$\lim_{t \to \infty}  \tilde{y}_k = 0.$$
To this end, we consider that the controller represents virtual couplings between the agents; see \cite{van2014port} and \cite{duindam2009modeling}. Specifically, virtual springs determine the formation shape, while virtual dampers inject damping, shaping the transient response. To illustrate this idea, we consider the link $k$ and the Hamiltonian (energy) function 
\begin{equation*}
    H_k(\tilde{y}_k):=\frac{1}{2}\tilde{y}_k^Tc_k\tilde{y}_k,
\end{equation*}
 where $c_k$ represents the constant stiffness matrix associated with the virtual springs. Note that the dimension of $c_k$ is determined by the constraint $y_k$. Hence, we propose the following dynamical system
 \begin{equation} \label{ctld}
\begin{array}{rcl}
     \dot{\tilde{y}}_k&=&\omega_k, \\
     \gamma_k&=&\frac{\partial{H_k}}{\partial{\tilde{y}_k}}+d_k\omega_k,
\end{array}
\end{equation} 
where the velocity of the constraint $\omega_k\in T_{y_k}\mathbb{Y}_k$---with $T_{y_k}\mathbb{Y}_k$ the tangent space of $\mathbb{Y}_k$---can be understood as the input of the controller system, while the output $\gamma_k$ corresponds to the force exerted by the virtual couplings. In particular, $\frac{\partial{H_k}}{\partial{\tilde{y}_k}}$ represents the forces due to the virtual springs, and $d_k\omega_k$ the forces due to the virtual dampers, where $d_k \geq 0$ correspond to a constant dissipation matrix.
 

Note that the dynamics of the virtual couplings are expressed in the constraint space. However, to implement the controller, they can be transformed to $\mathbb{R}^2$ by means of the constraint Jacobian, i.e.,
\begin{equation}	
\epsilon_k=L_{y_k}^T\gamma_k,
\end{equation}  
where $\epsilon_k$ is the virtual force in $\mathbb{R}^2$.

According to port-Hamiltonian theory, we define $<e|f>:=e^Tf$ as the power of a port with an effort $e$ and a flow $f$. Regarding the controller (\ref{ctld}), $\gamma_k$ is the effort, and $\dot{\tilde{y}}_k$ is the flow in the constraint space. Correspondingly, $\epsilon_k$ and $\dot{z}_k$ are the effort and flow in $\mathbb{R}^2$, respectively. Note that the power in different spaces is the same since energy is coordinate-free.  

\subsection{Problem Formulation}

Consider a group of $N$ agents with the topology of information exchange between these agents described by a graph $\mathcal{G}(\mathcal{V}_N,\mathcal{E}_E)$. It consists of a node set $\mathcal{V}$, where $\mathcal{V}=\{n_1,n_2,...,n_N\}$, and an edge set $\mathcal{E}\subseteq \mathcal{V} \times \mathcal{V}$, where $\mathcal{E}=\{e_1,e_2,...,e_E\}$. The incidence matrix $B \in \mathbb{R}^{N\times E}$ describes the relationship between the nodes and the edges, and it takes the following form:
$$
b_{ik}=
\begin{cases}
+1 & \text{if node $i$ is at the positive side of edge $k$},\\
-1 & \text{if node $i$ is at the negative side of edge $k$},\\
0  & \text{otherwise}.
\end{cases}
$$
In our work, we consider the particular class of strongly nondegenerate triangulated Laman graphs introduced in \cite{chen2017global}. Moreover, we consider that the interaction topology is determined only by angles. A triangulated Laman graph is constructed from a line graph with two nodes; then, every newly added node is connected by two existing nodes which are also connected. An example is shown in Fig. 2 with $M$ triangles. The formation shape is determined by the marked angles. Note that assuming the underlying graph as a triangulated Laman graph is not restrictive. Any angle-based formation shape can be designed as a triangulated Laman graph with a new set of angle constraints.

Note that the graph considered here is undirected and the angle rigidity is ensured by Lemma 1 in \cite{jing2019angle}, which is repeated below for the sake of completeness.

\begin{figure}[!ht]
\begin{center}
\label{Lamangraph_m}
\includegraphics[width=4.0cm]{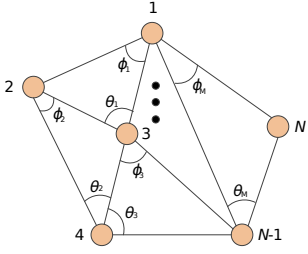}   
\caption{Triangulated Laman graph with $M$ triangles} 
\end{center}
\end{figure}

\begin{lemma} 
A triangulated Laman framework $(\mathcal{G}_N,q)$ is strongly nondegenerate only if $(\mathcal{G}_N,q)$ is globally angle rigid.
\end{lemma}

A framework is composed of a graph $\mathcal{G}_N$ and a realization which are positions of a group of agents, $q$. Angle rigidity means that the formation geometric shape is uniquely determined by a set of angles up to translations, scaling, and rotations by inner angles. For more details, see \cite{chen2020angle} and \cite{jing2019angle}. 

\begin{remark}
Assuming the underlying graph as a triangulated Laman graph with $M$ triangles is not restrictive, since any angle-based formation shape can be designed as a triangulated Laman graph with a new set of angle constraints. We remark that $2M$ is the least number of angle constraints to guarantee the angle rigidity. 
\end{remark}

We consider that the dynamics of the agents are given by double integrators with damping terms in $\mathbb{R}^2$, which are expressed in pH form as
\begin{equation} \label{model}	
\begin{array}{rcl}
\matr{\dot{q}_n  \\ \dot{p}_n } &=& \matr{0 & I_2  \\ -I_2 & -D_n^r }  \matr{\frac{\partial{H_n}}{\partial{q_n}}  \\[0.3cm]  \frac{\partial{H_n}}{\partial{p_n}}} + \matr{ 0  \\ I_2}U_n,\\[0.8cm]
H_n(p_{n})&=& \frac{1}{2m_n}p_n^Tp_n, \quad
Y_n=\frac{\partial{H_n}}{\partial{p_n}}(p_{n}),
\end{array}
\end{equation} 
where $q_n=(q_{x_n},q_{y_n})^T \in \mathbb{R}^2$, $n\in \{1,2,3\}$ denotes the position of the agent $n$, $p_n=(p_{x_n},p_{y_n})^T=(m_n\dot{q}_{x_n},m_n\dot{q}_{y_n})^T \in \mathbb{R}^2$ denotes the momentum, and $m_n$ denotes the mass. Moreover, $U_n=(U_{x_n},U_{y_n})^T \in \mathbb{R}^2$ and $Y_n=(Y_{x_n},Y_{y_n})^T \in \mathbb{R}^2$ denote the input and output, respectively, and $H_n(p_n)$ is the Hamiltonian of the agent $n$.  

We assume that each agent has access to bearing and relative velocity information. For relative velocity information, each agent is either able to measure and communicate its own velocity or measure the relative velocity directly. In addition, we assume that the communication topology is connected in each triangle, therefore, the agent has access to the bearing measurement of the non-adjacent edge. For example, as shown in Fig. 1, agent $1$ has access to the bearing measurement $s_i$. Since the constraints of the formation geometric shape are given by angles and agents can measure the bearings of their neighboring agents, we use the cosine function of the angle to represent the value of the angle, which can be easily calculated by bearing measurements. Hence, for the angle $\theta$, the definition is given by 
\begin{equation}	
\cos \theta:= s_k^Ts_j. \label{costheta}
\end{equation} 
The objective of this paper is to design a controller using only bearing and velocity measurements that ensures the group of agents with an underlying triangulated Laman graph, modeled by (\ref{model}) to achieve a desired formation constrained by angles, while simultaneously achieving velocity tracking, and scale and orientation control.

\section{Formation Stabilization}

Every triangulated Laman graph consists of several triangular graphs. Hence, we first present the controller for a specific triangular formation as the one shown in Fig. 1, where $1, 2, 3$ are agents and $i, j, k$ are the edges. The angles $\theta$ and $\phi$ are the ones to be controlled. Since the summation of three angles equals $\pi$, we only need to control $\theta$ and $\phi$.

From the definition \eqref{costheta}, the time-evolution of $\cos\theta$ can be derived as 
\begin{equation}
\begin{split}
\dot{\widetilde{(s_k^Ts_j)}}&=(L_{s_k}\dot{z}_k)^Ts_j+s_k^T(L_{s_j}\dot{z}_j)
              \\&=-(s_j^TL_{s_k}+s_k^TL_{s_j})\dot{q}_1+s_j^TL_{s_k}\dot{q}_2+s_k^TL_{s_j}\dot{q}_3
              \\&=L_{\theta 1}\dot{q}_1+L_{\theta 2}\dot{q}_2+L_{\theta 3}\dot{q}_3.
\end{split}
\end{equation}

Here, $L_{\theta 1},L_{\theta 2},L_{\theta 3}$ are the angle Jacobians mapping from the positions of agents $1$,$2$, and $3$, respectively, to $\theta$. Such angle Jacobians are given as follows
\begin{equation} \label{l1}
\begin{split}
&L_{\theta 1}=-(s_j^TL_{s_k}+s_k^TL_{s_j}),\\
&L_{\theta 2}=s_j^TL_{s_k},\\
&L_{\theta 3}=s_k^TL_{s_j}.
\end{split}
\end{equation}
Similarly, the angle Jacobians mapping from the positions of agents $1$, $2$, and $3$, respectively, to $\phi$ are given as
\begin{equation} \label{l2}
\begin{split}
&L_{\phi 1}=-s_i^TL_{s_j},\\
&L_{\phi 2}=-s_j^TL_{s_i},\\
&L_{\phi 3}=s_j^TL_{s_i}+s_i^TL_{s_j}.
\end{split}
\end{equation}

We can conclude from the Jacobians that the control law of each agent contains two parts: one to steer $\theta$ and the other to steer $\phi$.

Now we consider triangulated Laman graph with $M$ triangles. The angle Jacobian of each pair of angles $(\theta_l, \phi_l), l \in \{1, 2, ..., M\}$ is related to the corresponding three agents forming triangle $l$. The compact form is derived as

\begin{equation} \label{arm}
\begin{split}
\bordermatrix{
         & 1 &\cdots & n_1 &\cdots & n_2 &\cdots & n_3 &\cdots & N \cr 
\theta_1 &\cdots &\cdots &\cdots &\cdots &\cdots &\cdots &\cdots &\cdots &\cdots \cr
\phi_1   &\cdots &\cdots &\cdots &\cdots &\cdots &\cdots &\cdots &\cdots &\cdots \cr
\cdots   &\cdots &\cdots &\cdots &\cdots &\cdots &\cdots &\cdots &\cdots &\cdots \cr
\theta_l &\cdots &\cdots &L_{\theta_l n_1} &\cdots &L_{\theta_l n_2} &\cdots &L_{\theta_l n_3} &\cdots &\cdots \cr
\phi_l   &\cdots &\cdots &L_{\phi_l n_1} &\cdots &L_{\phi_l n_2} &\cdots &L_{\phi_l n_3} &\cdots &\cdots \cr
\cdots   &\cdots &\cdots &\cdots &\cdots &\cdots &\cdots &\cdots &\cdots &\cdots \cr
\theta_M   &\cdots &\cdots &\cdots &\cdots &\cdots &\cdots &\cdots &\cdots &\cdots \cr
\phi_M   &\cdots &\cdots &\cdots &\cdots &\cdots &\cdots &\cdots &\cdots &\cdots  
}.
\end{split}
\end{equation}

\subsection{Controller of agent 1 for $\theta$}

We first consider the controller of agent $1$. Since the movement of agent $1$ affects both $\theta$ and $\phi$, the controller of agent $1$ consists of two parts. One is to satisfy the constraint on $\theta$, and the other satisfies the constraint on $\phi$.

For $ \theta$, we use the cosine function to represent the angle measurement. Given \eqref{costheta}, the control aim is to design a controller to ensure that $s_k^Ts_j$ converge to the desired value $({s_k^Ts_j})^*$. Hence, we define the error
\begin{equation} \label{err1}
\begin{split}
\widetilde{(s_k^Ts_j)}:=(s_k^Ts_j)-({s_k^Ts_j})^*.
\end{split}
\end{equation}

In order to guarantee that the system converges to the desired state, i.e., the error (\ref{err1}) converges to zero, it is necessary to assign a potential energy of the angle to the closed-loop system. To this end, we propose the following Hamiltonian function
\begin{equation}
\begin{split}
H_{\theta 1}=\frac{1}{2}c_{\theta 1}\widetilde{(s_k^Ts_j)}^2,
\end{split}
\end{equation}
where $c_{\theta 1} \in \mathbb{R}_+$ is a constant.

The corresponding port-Hamiltonian controller composed of a spring and a damper is given as
\begin{equation} \label{ga1a}
\begin{array}{rcl}
\dot{\widetilde{(s_k^Ts_j)}}&=&\omega_{\theta 1},\\[0.2cm]
\gamma_{\theta 1}&=&\frac{\partial{H_{\theta 1}}}{\partial{\widetilde{(s_k^Ts_j)}}}+d_{\theta 1}\omega_{\theta 1},
\end{array}
\end{equation}
where $\omega_{\theta 1}$ denotes the input of the controller, and $d_{\theta 1}\in \mathbb{R}_+$ is the corresponding damping parameter. Here, $\gamma_{\theta 1}$ is the resulting virtual force in the space of angle constraint. Furthermore, the power of the port is given as
\begin{equation}
\begin{split}
 <\gamma_{\theta 1}|\dot{\widetilde{(s_k^Ts_j)}}>=\gamma_{\theta 1}^T\dot{\widetilde{(s_k^Ts_j)}}.
\end{split}
\end{equation}
where $\dot{\widetilde{(s_k^Ts_j)}}$ is the flow and $\gamma_{\theta 1}$ is the effort. 
Here we only consider a part of the time derivative of the angle $\theta$ caused by the movement of agent $1$. To transform the power from the angle space to $\mathbb{R}^2$, we compute
\begin{equation} \label{transp}
\begin{split}
<\gamma_{\theta 1}|\dot{\widetilde{(s_k^Ts_j)}}>&=<\gamma_{\theta 1}|L_{\theta 1}\dot{q}_1>\\&=<L_{\theta 1}^T\gamma_{\theta 1}|\dot{q}_1>\\&=<-L_{s_k}^Ts_j\gamma_{\theta 1}|\dot{q}_1>+<-L_{s_j}^Ts_k\gamma_{\theta 1}|\dot{q}_1>.
\end{split}
\end{equation}

The effort of the port in (\ref{transp}) relies on the distance information which is not measurable. In order to avoid distance measurement, we use the relative velocity measurement to estimate the unknown distance.

Note that since we assume that the distance is estimated rather than measured, the distance term in the angle Jacobian needs to be replaced accordingly. As a result, the estimated angle Jacobian is given by
\begin{equation} \label{eL1a}
\begin{split}
\hat{L}_{\theta 1}&=-(s_j^T\hat{L}_{\theta s_k}+s_k^T\hat{L}_{\theta s_j})\\&=-s_j^T\frac{1}{\hat{r}_{\theta k}}(I_2-s_k{s_k}^T)-s_k^T\frac{1}{\hat{r}_{\theta j}}(I_2-s_j{s_j}^T),
\end{split}
\end{equation}
where $\hat{r}_{\theta k}$ is the distance estimate of the edge $k$ for the control of angle $\theta$, and $\hat{r}_{\theta j}$ is the distance estimate of edge $j$ for the control of angle $\theta$. Correspondingly, $\hat{L}_{\theta s_k},\hat{L}_{\theta s_j}$ are the estimated bearing Jacobian of $s_k$ and $s_j$, respectively, for the control of angle $\theta$.

However if $\hat{L}_{\theta 1}$ is used to replace $L_{\theta 1}$ in the right side of (\ref{transp}), the equation is not satisfied because the effort $\gamma_{\theta 1}$ corresponds to the real flow $\frac{{\rm d}(s_k^Ts_j)}{{\rm d} t}$ in the angle space. It causes the discrepancy of the power through the virtual coupling due to the error between the estimated distance and the real unknown distance. To handle the discrepancies, we define the distance errors as $\bar{r}_{\theta k}:=\hat{r}_{\theta k}-r_k$, and $\bar{r}_{\theta j}:=\hat{r}_{\theta j}-r_j$. To calculate the estimated effort in $\mathbb{R}^2$, we have
\begin{equation} \label{alpha}
\begin{split}
-(\hat{L}_{\theta s_k}^Ts_j+\hat{L}_{\theta s_j}^Ts_k)\gamma_{\theta 1}=-L_{s_k}^Ts_j\alpha_{\theta k}-L_{s_j}^Ts_k\alpha_{\theta j},
\end{split}
\end{equation}
where $\alpha_{\theta k}=\frac{r_k}{\hat{r}_{\theta k}}\gamma_{\theta 1}$ is the estimated effort related to $\hat{r}_{\theta k}$ and $\alpha_{\theta j}=\frac{r_j}{\hat{r}_{\theta j}}\gamma_{\theta 1}$ is the estimated effort related to $\hat{r}_{\theta j}$. 

Furthermore, considering the ports in different spaces, we have that 
\begin{equation} \label{power1}
\begin{split}
<\hat{L}_{\theta 1}^T\gamma_{\theta 1}|\dot{q}_1>&=<-L_{s_k}^Ts_j\alpha_{\theta k}-L_{s_j}^Ts_k\alpha_{\theta j}|\dot{q}_1>\\& =<\alpha_{\theta k}|\dot{\widetilde{(s_k^Ts_j)}}>+<\alpha_{\theta j}|\dot{\widetilde{(s_k^Ts_j)}}>.
\end{split}
\end{equation}

Comparing (\ref{transp}) with (\ref{power1}), the discrepancy between the real effort and the estimated effort can be derived as
\begin{equation} \label{beta}
\begin{split}
\beta_{\theta k}=\alpha_{\theta k}-\gamma_{\theta 1}=-\frac{\bar{r}_k}{\hat{r}_{\theta k}}\gamma_{\theta 1},\\
\beta_{\theta j}=\alpha_{\theta j}-\gamma_{\theta 1}=-\frac{\bar{r}_j}{\hat{r}_{\theta j}}\gamma_{\theta 1}.
\end{split}
\end{equation}

The power of ports with $\beta_{\theta k}$, $\beta_{\theta j}$ as the efforts are given by 
\begin{equation} \label{aport}
\begin{split}
<\beta_{\theta k}|-s_j^TL_{s_k}\dot{q}_1>,\\\quad <\beta_{\theta j}|-s_k^TL_{s_j}\dot{q}_1>. 
\end{split}
\end{equation}

To account for the power associated with the ports in distance space, we define the corresponding Hamiltonians as
\begin{equation} \label{Hbeta}
\begin{split}
H_{\theta k}:=\frac{1}{2}c_{\theta k}\bar{r}_{\theta k}^2,\\
H_{\theta j}:=\frac{1}{2}c_{\theta j}\bar{r}_{\theta j}^2,
\end{split}
\end{equation}
where $c_{\theta k},c_{\theta j} \in \mathbb{R}_+$ are constants. 

The power of the ports in distance space is given by
\begin{equation} \label{dport}
\begin{split}
<\frac{\partial{H_{\theta k}}}{\partial{\bar{r}_{\theta k}}}|\dot{\bar{r}}_{\theta k}>=<c_{\theta k}\bar{r}_{\theta k}|\dot{\bar{r}}_{\theta k}>,\\
<\frac{\partial{H_{\theta j}}}{\partial{\bar{r}_{\theta j}}}|\dot{\bar{r}}_{\theta j}>=<c_{\theta j}\bar{r}_{\theta j}|\dot{\bar{r}}_{\theta j}>.
\end{split}
\end{equation}

Since the energy is coordinate-free, the power in angle space and distance space are the same. Therefore, comparing (\ref{aport}) and (\ref{dport}), we have 
\begin{equation} \label{br1k}
\begin{array}{rrcl}
     &<\beta_{\theta k}|-s_j^TL_{s_k}\dot{q}_1>&=&<c_{\theta k}\bar{r}_{\theta k}|\dot{\bar{r}}_{\theta k}>\\
     \Rightarrow&c_{\theta k}\bar{r}_{\theta k}^T\dot{\bar{r}}_{\theta k}&=&-\frac{\bar{r}_k}{\hat{r}_{\theta k}}\gamma_{\theta 1}^T(-s_j^TL_{s_k}\dot{q}_1)\\
     \Rightarrow&\dot{\bar{r}}_{\theta k}&=&-\frac{\gamma_{\theta 1}^T}{c_{\theta k}\hat{r}_{\theta k}}(-s_j^TL_{s_k}\dot{q}_1).
\end{array}
\end{equation}

Similarly, 
\begin{equation} \label{br1j}
\begin{split}
\dot{\bar{r}}_{\theta j}=-\frac{\gamma_{\theta 1}^T}{c_{\theta j}\hat{r}_{\theta j}}(-s_k^TL_{s_j}\dot{q}_1).
\end{split}
\end{equation}

Furthermore, the dynamics of the estimators are given by
\begin{equation} \label{er1k}
\begin{split}
\dot{\hat{r}}_{\theta k}&=\dot{r}_k+\dot{\bar{r}}_{\theta k}\\
&=s_k^T\dot{z}_k-\frac{\gamma_{\theta 1}^T}{c_{\theta k}\hat{r}_{\theta k}}(-s_j^TL_{s_k}\dot{q}_1).
\end{split}
\end{equation}
Similarly,
\begin{equation} \label{er1j}
\begin{split}
\dot{\hat{r}}_{\theta j}&=\dot{r}_j+\dot{\bar{r}}_{\theta j}\\
&=s_j^T\dot{z}_j-\frac{\gamma_{\theta 1}^T}{c_{\theta j}\hat{r}_{\theta j}}(-s_k^TL_{s_j}\dot{q}_1).
\end{split}
\end{equation}

Note that we only use the information of relative velocity and bearing measurement in the above estimators, while distance measurement is not used.

The resulting control law of agent $1$ for angle $\theta$ is given by feedback connection of the original system and the effort in $\mathbb{R}^2$ with the estimated distance, which can be summarized by substituting (\ref{ga1a}) and (\ref{eL1a}) into (\ref{alpha}) and adding a negative sign as
\begin{equation} 
\begin{split}
U_{\theta 1}=&-\overbrace{[-\frac{1}{\hat{r}_{\theta k}}(I_2-s_k{s_k}^T)^Ts_j-\frac{1}{\hat{r}_{\theta j}}(I_2-s_j{s_j}^T)^Ts_k]}^{\hat{L}_{\theta 1}}\\&
\underbrace{[c_{\theta 1}{{\widetilde{(s_k^Ts_j)}}}+d_{\theta 1}\dot{\widetilde{(s_k^Ts_j)}}]}_{\gamma_{\theta 1}}.
\end{split}
\end{equation}

\subsection{Controller of agent $1$ for $\phi$}
In this section, we design the controller of agent $1$ for the control of angle $\phi$. Define the Hamiltonian
\begin{equation}
\begin{split}
H_{\phi 1}:=\frac{1}{2}c_{\phi 1}\widetilde{(s_i^Ts_j)}^2,
\end{split}
\end{equation}
where $c_{\phi 1} \in \mathbb{R}_+$ is a constant. The controller with spring and damping terms are given by
\begin{equation}
\begin{array}{rcl}
    \dot{\widetilde{(s_i^Ts_j)}}&=&\omega_{\phi 1},\\
\gamma_{\phi 1}&=&\frac{\partial{H_{\phi 1}}}{\partial{\widetilde{(s_i^Ts_j)}}}+d_{\phi 1}\omega_{\phi 1},
\end{array}
\end{equation}
where $\omega_{\phi 1}$ denotes the input of the controller. $d_{\phi 1} \in \mathbb{R}_+$ is the corresponding damping parameter.

Considering the ports in different spaces, we have that 
\begin{equation} \label{eL2a}
\begin{split}
<\alpha_{\phi 1}|(-s_i^TL_{s_j}\dot{q}_1)>=<\hat{L}_{\phi 1}^T\gamma_{\phi 1}|\dot{q}_1>,\\\hat{L}_{\phi 1}=s_i^T\hat{L}_{\phi s_k}=s_i^T\frac{1}{\hat{r}_{\phi k}}(I_2-s_k{s_k}^T),
\end{split}
\end{equation}
where $\hat{L}_{\phi 1}$ is the estimated angle Jacobian mapping from the position of the agent $1$ to the angle $\phi$ and $\hat{L}_{\phi s_k}$ is the estimated bearing Jacobian of $s_k$ for the control of angle $\phi$. Furthermore, $\hat{r}_{\phi k}$ is the estimated distance of the edge $k$ for the control of angle $\phi$.

Taking the same steps as in Section 3.1 leads to the following estimator
\begin{equation} \label{er2ka}
\begin{split}
\dot{\hat{r}}_{\phi k}=\dot{r}_{k}+\dot{\bar{r}}_{\phi k}=s_k^T\dot{z}_k-\frac{\gamma_{\phi 1}^T}{c_{\phi k}\hat{r}_{\phi k}}(-s_i^TL_{s_j}\dot{q}_1).
\end{split}
\end{equation}
where $c_{\phi k}  \in \mathbb{R}_+$ is a constant.

Correspondingly, the control law of the agent $1$ for the angle $\phi$ is given by
\begin{equation} 
\begin{split}
U_{\phi 1}=-\underbrace{\frac{1}{\hat{r}_{\phi k}}(I_2-s_k{s_k}^T)^Ts_i}_{\hat{L}_{\phi 1}}\underbrace{[c_{\phi 1}{{\widetilde{(s_i^Ts_j)}}}+d_{\phi 1}\dot{\widetilde{(s_i^Ts_j)}}]}_{\gamma_{\phi 1}}.
\end{split}
\end{equation}


To sum up, the control law of the agent $1$ is given by
\begin{equation} \label{ctra}
\begin{split}
U_1&=U_{\theta 1}+U_{\phi 1}\\&-[\frac{1}{\hat{r}_{\theta k}}(I_2-s_k{s_k}^T)^Ts_j+\frac{1}{\hat{r}_{\theta j}}(I_2-s_j{s_j}^T)^Ts_k]\\&
[c_{\theta 1}{{\widetilde{(s_k^Ts_j)}}}+d_{\theta 1}\dot{\widetilde{(s_k^Ts_j)}}]
\\&-\frac{1}{\hat{r}_{\phi k}}(I_2-s_k{s_k}^T)^Ts_i[c_{\phi 1}{{\widetilde{(s_i^Ts_j)}}}+d_{\phi 1}\dot{\widetilde{(s_i^Ts_j)}}].
\end{split}
\end{equation}

\subsection{Controllers for the agents $2$ and $3$}
The controller design for the agents $2$ and $3$ is similar to the agent $1$, so we omit some details for the following analysis.

For the agent $2$, there are two parts in the control law. The first part is to control the angle $\theta$, whose corresponding Hamiltonian and controller in the angle space are given by
\begin{equation}
\begin{split}
H_{\theta 2}=\frac{1}{2}c_{\theta 2}\widetilde{(s_k^Ts_j)}^2,
\end{split}
\end{equation}
\begin{equation} \label{ga1b}
\begin{split}
\dot{\widetilde{(s_k^Ts_j)}}=\omega_{\theta 2},\\\gamma_{\theta 2}=\frac{\partial{H_{\theta 2}}}{\partial{\widetilde{(s_k^Ts_j)}}}+d_{\theta 2}\omega_{\theta 2},
\end{split}
\end{equation}
where $\omega_{\theta 2}$ denotes the input of the controller. $c_{\theta 2},d_{\theta 2}  \in \mathbb{R}_+$ are constants. The corresponding estimated angle Jacobian and distance estimators are given by
\begin{equation} \label{eL1b}
\begin{split}
\hat{L}_{\theta 2}=s_j^T\frac{1}{\hat{r}_{\theta i}}(I_2-s_i{s_i}^T),
\end{split}
\end{equation}
\begin{equation} \label{er1ib}
\begin{split}
\dot{\hat{r}}_{\theta i}=\dot{r}_{i}+\dot{\bar{r}}_{\theta i}=s_i^T\dot{z_i}-\frac{\gamma_{\theta 2}^T}{c_{\theta i}\hat{r}_{\theta i}}(s_j^TL_{s_k}\dot{q}_2).
\end{split}
\end{equation}

The second part is to control the angle $\phi$, whose corresponding Hamiltonian and controller are given by
\begin{equation}
\begin{split}
H_{\phi 2}=\frac{1}{2}c_{\phi 2}\widetilde{(s_i^Ts_j)}^2,
\end{split}
\end{equation}
\begin{equation}
\begin{split}
\dot{\widetilde{(s_i^Ts_j)}}=\omega_{\phi 2},\\\gamma_{\phi 2}=\frac{\partial{H_{\phi 2}}}{\partial{\widetilde{(s_i^Ts_j)}}}+d_{\phi 2}\omega_{\phi 2},
\end{split}
\end{equation}
where $\omega_{\phi 2}$ denotes the input of the controller, and $c_{\phi 2},d_{\phi 2}  \in \mathbb{R}_+$ are constants. The corresponding estimated angle Jacobian and distance estimators are given by
\begin{equation} \label{eL2b}
\begin{split}
\hat{L}_{\phi 2}=s_j^T\frac{1}{\hat{r}_{\phi k}}(I_2-s_k{s_k}^T),
\end{split}
\end{equation}
\begin{equation} \label{er2kb}
\begin{split}
\dot{\hat{r}}_{\phi k}=\dot{r}_{k}+\dot{\bar{r}}_{\phi k}=s_k^T\dot{z}_k-\frac{\gamma_{\phi 2}^T}{c_{\phi k}\hat{r}_{\phi k}}(-s_j^TL_{s_i}\dot{q}_2).
\end{split}
\end{equation}

Note that the expressions (\ref{er2ka}) and (\ref{er2kb}) are both distance estimators of the edge $k$, but (\ref{er2ka}) is estimated by the agent $1$, while (\ref{er2kb}) is estimated by the agent $2$.

The power from the port related to agent $2$ is given by
\begin{equation} 
\begin{split}
<\hat{L}_{\theta 2}^T\gamma_{\theta 2}+\hat{L}_{\phi 2}^T\gamma_{\phi 2}|\dot{q}_2>.
\end{split}
\end{equation}
Correspondingly, the control law of the agent $2$ is given by
\begin{equation} \label{ctrb}
\begin{split}
U_{2}&=U_{\theta 2}+U_{\phi 2}\\&=-\overbrace{\frac{1}{\hat{r}_{\theta i}}(I_2-s_i{s_i}^T)^Ts_j}^{\hat{L}_{\theta 2}}
\overbrace{[c_{\theta 2}{{\widetilde{(s_k^Ts_j)}}}+d_{\theta 2}\dot{\widetilde{(s_k^Ts_j)}}]}^{\gamma_{\theta 2}}\\
&\quad-\underbrace{\frac{1}{\hat{r}_{\phi k}}(I_2-s_k{s_k}^T)^Ts_j}_{\hat{L}_{\phi 2}}\underbrace{[c_{\phi 2}{{\widetilde{(s_i^Ts_j)}}}+d_{\phi 2}\dot{\widetilde{(s_i^Ts_j)}}]}_{\gamma_{\phi 2}}.
\end{split}
\end{equation}

For the agent $3$, there are also two parts in the control law. The first part is to control the angle $\phi$, whose corresponding Hamiltonian and controller in angle space are given by
\begin{equation}
\begin{split}
H_{\phi 3}=\frac{1}{2}c_{\phi 3}\widetilde{(s_i^Ts_j)}^2,
\end{split}
\end{equation}
\begin{equation}
\begin{split}
\dot{\widetilde{(s_i^Ts_j)}}=\omega_{\phi 3},\\\gamma_{\phi 3}=\frac{\partial{H_{\phi 3}}}{\partial{\widetilde{(s_i^Ts_j)}}}+d_{\phi 3}\omega_{\phi 3},
\end{split}
\end{equation}
where $\omega_{\phi 3}$ denotes the input of the controller, and $c_{\phi 3},d_{\phi 3}  \in \mathbb{R}_+$ are constants. The corresponding estimated angle Jacobian and distance estimators are given
\begin{equation} \label{eL2c}
\begin{split}
\hat{L}_{\phi 3}=s_j^T\frac{1}{\hat{r}_{\phi i}}(I_2-s_i{s_i}^T)+s_i^T\frac{1}{\hat{r}_{\phi j}}(I_2-s_j{s_j}^T),
\end{split}
\end{equation}
\begin{equation} \label{er2i}
\begin{split}
\dot{\hat{r}}_{\phi i}=\dot{r}_{i}+\dot{\bar{r}}_{\phi i}=s_i^T\dot{z}_i-\frac{\gamma_{\phi 3}^T}{c_{\phi i}\hat{r}_{\phi i}}(s_j^TL_{s_i}\dot{q}_3),
\end{split}
\end{equation}
\begin{equation} \label{er2j}
\begin{split}
\dot{\hat{r}}_{\phi j}=\dot{r}_{j}+\dot{\bar{r}}_{\phi j}=s_j^T\dot{z}_j-\frac{\gamma_{\phi 3}^T}{c_{\phi j}\hat{r}_{\phi j}}(s_i^TL_{s_j}\dot{q}_3),
\end{split}
\end{equation}
where $c_{\phi i},c_{\phi j}  \in \mathbb{R}_+$ are constants.

The second part is to control the angle $\theta$. The corresponding Hamiltonian and controller in angle space are given by
\begin{equation}
\begin{split}
H_{\theta 3}=\frac{1}{2}c_{\theta 3}\widetilde{(s_k^Ts_j)}^2,
\end{split}
\end{equation}
\begin{equation} \label{ga1c}
\begin{split} 
\dot{\widetilde{(s_k^Ts_j)}}=\omega_{\theta 3},\\\gamma_{\theta 3}=\frac{\partial{H_{\theta 3}}}{\partial{\widetilde{(s_k^Ts_j)}}}+d_{\theta 3}\omega_{\theta 3},
\end{split}
\end{equation}
where $\omega_{\theta 3}$ denotes the input of the controller, and $c_{\theta 3},d_{\theta 3} \in \mathbb{R}_+$ are constants. The corresponding estimated angle Jacobian and distance estimators are given by
\begin{equation} \label{eL1c}
\begin{split}
\hat{L}_{\theta 3}=s_k^T\frac{1}{\hat{r}_{\theta i}}(I_2-s_i{s_i}^T),
\end{split}
\end{equation}
\begin{equation} \label{er1ic}
\begin{split}
\dot{\hat{r}}_{\theta i}=\dot{r}_{i}+\dot{\bar{r}}_{\theta i}=s_i^T\dot{z_i}-\frac{\gamma_{\theta 3}^T}{c_{\theta i}\hat{r}_{\theta i}}(s_k^TL_{s_j}\dot{q}_3),
\end{split}
\end{equation}
where $c_{\theta i}  \in \mathbb{R}_+$ is a constant. Note that the expressions (\ref{er1ib}) and (\ref{er1ic}) are both distance estimators of the edge $i$, but (\ref{er1ib}) is estimated by the agent $2$, while (\ref{er1ic}) is estimated by the agent $3$.

The power from the port related to agent $3$ is given by
\begin{equation} 
\begin{split}
<\hat{L}_{\phi 3}^T\gamma_{\phi 3}+\hat{L}_{\theta 3}^T\gamma_{\theta 3}|\dot{q}_3>.
\end{split}
\end{equation}
Correspondingly, the controller of the agent $3$ is given by
\begin{equation} \label{ctrc}
\begin{split}
U_{3}=&U_{\theta 3}+U_{\phi 3}\\=&-\overbrace{[\frac{1}{\hat{r}_{\phi i}}(I_2-s_i{s_i}^T)^Ts_j+{\frac{1}{\hat{r}_{\phi j}}}(I_2-s_j{s_j}^T)^Ts_i]}^{\hat{L}_{\phi 3}}\\& \overbrace{[c_{\phi 3}{\widetilde{(s_i^Ts_j)}}+d_{\phi 3}{\dot{\widetilde{(s_i^Ts_j)}}}]}^{\gamma_{\phi 3}}\\&-\underbrace{{\frac{1}{\hat{r}_{\theta i}} (I_2-s_i{s_i}^T)^Ts_k}}_{\hat{L}_{\theta 3}} \underbrace{[c_{\theta 3}{{\widetilde{(s_k^Ts_j)}}}+d_{\theta 3}{\dot{\widetilde{(s_k^Ts_j)}}}]}_{\gamma_{\theta 3}}.
\end{split}
\end{equation}

In general, the control law for three agents is derived as
\begin{equation} \label{ui}
\begin{split}
U_n&=-(\hat{L}_{\theta n}^T\gamma_{\theta n}+\hat{L}_{\phi n}^T\gamma_{\phi n})\\&=-(\hat{L}_{\theta n}^T(c_{\theta n}\widetilde{(s_k^Ts_j)}+d_{\theta n}\dot{\widetilde{(s_k^Ts_j)}})+\\&\quad \hat{L}_{\phi n}^T(c_{\phi n}\widetilde{(s_i^Ts_j)}+d_{\phi n}\dot{\widetilde{(s_i^Ts_j)}})), \quad n \in \{1, 2, 3\} .
\end{split}
\end{equation}

\subsection{Stability Analysis}
To facilitate the analysis, we first give the following lemma.
\begin{lemma} \label{lemma2}
Consider the three agents modeled as (\ref{model}). The desired formation are given by two desired angles $\theta^*$ and $\phi^*$ which are not close to $0$ or $\pi$. If the three agents are moving in the neighborhood of the desired formation shape, the matrix
\begin{equation} \label{Lmatr}
 \hat{L}^T:=\matr{
 \hat{L}^T_{\theta 1} & \hat{L}^T_{\phi 1} \\ \hat{L}^T_{\theta 2} & \hat{L}^T_{\phi 2} \\ 
 \hat{L}^T_{\theta 3} & \hat{L}^T_{\phi 3} }    \in \mathbb{R}^{6 \times 2}
\end{equation}
is of full column rank.
\end{lemma}

\begin{proof}
According to the definitions in (\ref{eL1a}), (\ref{eL2a}), (\ref{eL1b}), (\ref{eL2b}), (\ref{eL1c}) and (\ref{eL2c}), we have 
\begin{equation}
 \hat{L}^T=\matr{
 -(\hat{L}^T_{s_k}s_j+\hat{L}^T_{s_j}s_k) & -\hat{L}^T_{s_i}s_j \\ \hat{L}^T_{s_k}s_j & -\hat{L}^T_{s_j}s_i \\ 
 \hat{L}^T_{s_j}s_k & \hat{L}^T_{s_i}s_j+\hat{L}^T_{s_j}s_i
 }.    
\end{equation}
Note that all the entries in $\hat{L}^T$ are column vectors. We first consider the term ${L}^T_{s_k}s_j$ without introducing its estimator. ${L}^T_{s_k}$ is the orthogonal projection on $s_k$. A vector multiplied by ${L}^T_{s_k}$ makes the resulting vector parallel to ${s_k}^{\bot}$. Furthermore, the distance discrepancy caused by introducing the estimator only changes the magnitude of the vector. Therefore, the term $\hat{L}^T_{s_k}s_j$ can be rewritten as $a_1s_k^{\bot}$ with an unspecified magnitude parameter $a_1 \in \mathbb{R}_+$. Note that the value of $a_1$ does not affect the rank of $\hat{L}^T$. Using this property for each entry, the matrix can be rewritten as
\begin{equation}
\begin{split}
 \hat{L}^T=\matr{
 -(a_1s_k^{\bot}+a_2s_j^{\bot}) & -b_1s_i^{\bot} \\ a_1s_k^{\bot} & -b_2s_j^{\bot} \\ 
 a_2s_j^{\bot} & b_1s_i^{\bot}+b_2s_j^{\bot}
 },    
\end{split}
\end{equation}
where $a_1, a_2, b_1, b_2 \in \mathbb{R}_+$. Note that the values are determined by the current positions of three agents. Since $\theta^*$ and $\phi^*$ are not close to $0$ or $\pi$ in the neighborhood of the desired formation, any two of $s_i, s_j, s_k$ are linearly independent. Therefore, any two of $s_i^{\bot}, s_j^{\bot}, s_k^{\bot}$ are linearly independent. Furthermore, since $a_1, a_2, b_1, b_2$ are always positive, the two columns of $\hat{L}^T$ are always linearly independent, therefore, it is of full column rank. 
\end{proof}

The main result of this paper is given by the following theorem.
\begin{theorem} \label{thm1}
Consider the three agents modeled as in (\ref{model}). Assume that the initial positions of three agents are given by that the angles to be controlled are in the neighborhood of the desired formation. Using the control law (\ref{ctra}) for agent 1, the control law (\ref{ctrb}) for agent 2, and the control law (\ref{ctrc}) for agent 3, the three agents converge to the formation constrained by the desired angles.    
\end{theorem}

\begin{proof}
Take the following Hamiltonian as a candidate Lyapunov function 
\begin{equation} \label{H}
\begin{split}
H^t=&\frac{1}{2}(\frac{1}{m_1}p_1^Tp_1+\frac{1}{m_2}p_2^Tp_2+\frac{1}{m_3}p_3^Tp_3)\\&\frac{1}{2}(c_{\theta 1}+c_{\theta 2}+c_{\theta 3})\widetilde{(s_k^Ts_j)}^2+\\&\frac{1}{2}(c_{\phi 1}+c_{\phi 2}+c_{\phi 3})\widetilde{(s_i^Ts_j)}^2+\\&\frac{1}{2}c_{\theta i}\bar{r}_{\theta i2}^2+\frac{1}{2}c_{\theta i}\bar{r}_{\theta i3}^2+\frac{1}{2}c_{\theta j}\bar{r}_{\theta j}^2+\frac{1}{2}c_{\theta k}\bar{r}_{\theta k}^2+\\&\frac{1}{2}c_{\phi k}\bar{r}_{\phi k1}^2+\frac{1}{2}c_{2k}\bar{r}_{\phi k2}^2+\frac{1}{2}c_{\phi i}\bar{r}_{\phi i}^2+\frac{1}{2}c_{\phi j}\bar{r}_{\phi j}^2.
\end{split}
\end{equation}

It follows that $H^t$ is positive definite. Now we consider the time derivative of (\ref{H})
\begin{equation} \label{dH}
\begin{split}
\dot{H}^t=&m_1\dot{q}_1^T\ddot{q}_1+m_2\dot{q}_2^T\ddot{q}_2+m_3\dot{q}_3^T\ddot{q}_3\\&(c_{\theta 1}+c_{\theta 2}+c_{\theta 3})[\widetilde{(s_k^Ts_j)}(\hat{L}_{\theta 1}\dot{q}_1+\hat{L}_{\theta 2}\dot{q}_2+\hat{L}_{\theta 3}\dot{q}_3)]+\\&(c_{\phi 1}+c_{\phi 2}+c_{\phi 3})[\widetilde{(s_i^Ts_j)}(\hat{L}_{\phi 1}\dot{q}_1+\hat{L}_{\phi 2}\dot{q}_2+\hat{L}_{\phi 3}\dot{q}_3)]+\\&c_{\theta i}\bar{r}_{\theta i2}\dot{\bar{r}}_{\theta i}+c_{\theta i}\bar{r}_{\theta i3}\dot{\bar{r}}_{\theta i}+c_{\theta j}\bar{r}_{\theta j}\dot{\bar{r}}_{\theta j}+c_{\theta k}\bar{r}_{\theta k}\dot{\bar{r}}_{\theta k}+\\&c_{\phi k}\bar{r}_{\phi k1}\dot{\bar{r}}_{\phi k1}+c_{\phi k}\bar{r}_{\phi k2}\dot{\bar{r}}_{\phi k2}+c_{\phi i}\bar{r}_{\phi i}\dot{\bar{r}}_{\phi i}+c_{\phi j}\bar{r}_{\phi j}\dot{\bar{r}}_{\phi j}.
\end{split}
\end{equation}

Note that 
\begin{equation} \label{ddq}
\begin{split}
\dot{p}_n=-D_n^r\frac{p_n}{m_n}+U_n, \quad n \in \{1,2,3\}
\end{split}
\end{equation}
 
Substituting (\ref{ddq}) into the first line of (\ref{dH}), substituting (\ref{eL1a}), (\ref{eL1b}), (\ref{eL1c}) into the second line of (\ref{dH}), substituting (\ref{eL2a}), (\ref{eL2b}), (\ref{eL2c}) into the third line of (\ref{dH}), substituting (\ref{er1ib}), (\ref{er1ic}), (\ref{er1j}) and (\ref{er1k}) into the fourth line of (\ref{dH}), and substituting (\ref{er2kb}), (\ref{er2ka}), (\ref{er2i}) and (\ref{er2j}) into the fifth line of (\ref{dH}), we can simplify (\ref{dH}). For simplicity, we omit the process and only give the result as
\begin{equation} \label{dHf}
\begin{split}
\dot{H}^t=&-(D_1^r\dot q_1^T\dot q_1+D_2^r\dot q_2^T\dot q_2+D_3^r\dot q_3^T\dot q_3)\\&-(d_{\theta 1}+d_{\theta 2}+d_{\theta 3})\dot{\widetilde{(s_k^Ts_j)}}^2\\&-(d_{\phi 1}+d_{\phi 2}+d_{\phi 3})\dot{\widetilde{(s_i^Ts_j)}}^2\leq0.
\end{split}
\end{equation}

By invoking LaSalle's invariance principle, the trajectories of the closed-loop system converge to the largest invariant set where $\dot{\bar{H}}^t=0$. On this set $\dot{q}_1, \dot{q}_2, \dot{q}_3=0$, $\dot{\widetilde{(s_k^Ts_j)}}=0$ and $\dot{\widetilde{(s_i^Ts_j)}}=0$. Hence, $\ddot{q}_1, \ddot{q}_2, \ddot{q}_3=0$. Furthermore, according to (\ref{ddq}) it follows that 
\begin{equation} \label{Ueq0}
U_1, U_2, U_3=0
\end{equation}
on this invariant set.

Next we prove $\widetilde{(s_k^Ts_j)}=0$ and $\widetilde{(s_i^Ts_j)}=0$ on this invariant set. Substituting $\dot{\widetilde{(s_k^Ts_j)}}=0$, $\dot{\widetilde{(s_i^Ts_j)}}=0$, and (\ref{Ueq0}) into (\ref{ui}) for $n=1,2,3$ respectively, we have
\begin{equation} \label{sui}
\begin{split}
0=c_{\theta 1}\hat{L}_{\theta 1}^T\widetilde{(s_k^Ts_j)}+c_{\phi 1}\hat{L}_{\phi 1}^T\widetilde{(s_i^Ts_j)},\\0=c_{\theta 2}\hat{L}_{\theta 2}^T\widetilde{(s_k^Ts_j)}+c_{\phi 2}\hat{L}_{\phi 2}^T\widetilde{(s_i^Ts_j)},\\0=c_{\theta 3}\hat{L}_{\theta 3}^T\widetilde{(s_k^Ts_j)}+c_{\phi 3}\hat{L}_{\phi 3}^T\widetilde{(s_i^Ts_j)}.
\end{split}
\end{equation}
Note that $c_{\theta 1}, c_{\theta 2}, c_{\theta 3}, c_{\phi 1}, c_{\phi 2}, c_{\phi 3}$ are positive spring constants. Combined with the analysis in Lemma 2, the matrix $$\matr{
 c_{\theta 1}\hat{L}^T_{\theta 1} &  c_{\phi 1}\hat{L}^T_{\phi 1} \\ c_{\theta 2}\hat{L}^T_{\theta 2} & c_{\phi 2}\hat{L}^T_{\phi 2} \\ c_{\theta 3}\hat{L}^T_{\theta 3} & c_{\phi 3}\hat{L}^T_{\phi 3} }$$ is also of full column rank, thus we conclude that $\widetilde{(s_k^Ts_j)}=\widetilde{(s_i^Ts_j)}=0$, which means that the three agents achieve the desired formation, thus completing the proof. 
\end{proof}


\section{Extension to triangulated Laman graphs}

The topology of a group of agents can be designed according to the angle constraints. In this work we assume that the topology is designed as a triangulated Laman graph $\mathcal{G}_n(\mathcal{V}_n,\mathcal{E})$ with $M$ triangles as shown in Fig. 2. In each triangle, the form of the control law is designed as in Section 3. The corresponding control law for each agent depends on how many triangles the agent belongs to. The total control law is the sum of all control laws derived from all related triangles. We directly give the control law for the whole group as  
\begin{equation} \label{um}
\begin{split}
U_n &= \sum_{l \in \mathcal{N}_n}(U_{\theta_l n}+U_{\phi_l n})
\\&= \sum_{l \in \mathcal{N}_n}(\hat{L}_{\theta_l n}(c_{\theta_l n}\widetilde{\cos{\theta_l}}+d_{\theta_l n}\dot{\widetilde{\cos{\theta_l}}})
\\& \quad +\hat{L}_{\phi_l n}(c_{\phi_l n}\widetilde{\cos{\phi_l}}+d_{\phi_l n}\dot{\widetilde{\cos{\phi_l}}})), \quad n \in \{1,2,...,N\},
\end{split}
\end{equation}
where $l \in \{1,2,..., M\}$ represents the $l$th triangle and the set $\mathcal{N}_n$ contains all the triangles that agent $n$ forms. $\hat{L}_{\theta_l n}$ and $\hat{L}_{\phi_l n}$ are the estimated angle Jacobian mapping from angles $\theta_l$ and $\phi_l$ to agent $n$, respectively. The estimators in $\hat{L}_{\theta_l n}$ and $\hat{L}_{\phi_l n}$ are derived by the same approach as in Section 3.1-3.3 for each triangle. We omit the expressions due to their tedious and numerous variables and terms.

\begin{lemma} \label{lemma3}
Consider a triangulated Laman graph $\mathcal{G}_n(\mathcal{V}_n,\mathcal{E})$ with $M$ triangles, and all the agents are modeled as in (\ref{model}). The desired angle-based formation is given by the desired angles $\theta^*_l,\phi^*_l$ which are not close to $0$ or $\pi$. If the initial positions of all agents are in the neighborhood of the desired formation, the matrices
\begin{equation}
 \hat{L}_l^T=\matr{ 
 \hat{L}_{\theta_l}^T &  \hat{L}_{\phi_l}^T 
 } \in \mathbb{R}^{6 \times 2}, \quad  l\in \{1,2,...,M\}
\end{equation}
are of full column rank, where $\hat{L}_l^T$ is the Jacobian matrix mapping from the two angles of triangle $l$ to the three agents forming the triangle. 
\end{lemma}

\begin{proof}
$\hat{L}_l^T$ takes the same form as in (\ref{Lmatr}) . The first column $\hat{L}_{\theta l}^T$ is the Jacobian mapping from $\theta_l$ to the three agents forming the triangle $l$ and the second column $\hat{L}_{\phi l}^T$ is the Jacobian mapping from $\phi_l$ to the three agents. Note that all the desired angles $\theta^*_l,\phi^*_l, l \in \{1,2,...,M\}$ are not close to $0$ or $\pi$. Considering the $M$ triangles respectively and applying \textit{Lemma \ref{lemma3}}, we conclude that $\hat{L}_l$ is of full column rank for all $l \in \{1,2,...,M\}$ .
\end{proof}

\begin{theorem} \label{thm}
Consider a triangulated Laman graphs $\mathcal{G}_n(\mathcal{V}_n,\mathcal{E})$ with $M$ triangles. The angles to be controlled are defined by choosing any two angles of each triangle whose values are away enough from $0$ or $\pi$. Moreover, assume that the initial positions of all agents are given such that the angles to be controlled are in the neighborhood of the desired formation. Under the control law given by (\ref{um}), the group of agents locally achieves the desired angle-based formation.    
\end{theorem}

\begin{proof}
To analyze the stability, the general Hamiltonian is given as
\begin{equation} \label{gH}
\begin{split}
H^L=&\frac{1}{2}\sum_{n=1}^N \frac{1}{m_n}p_n^Tp_n
+ \frac{1}{2}\sum_{l=1}^M\sum_{n \in \mathcal{N}_l}(c_{\theta_l n}\widetilde{\cos{\theta_l}}^2  +c_{\phi_l n}\widetilde{\cos{\phi_l}}^2) \\&+\frac{1}{2}\sum_{\epsilon \in \mathcal{E}}c_{\epsilon}\bar{r}_{\epsilon}^2,
\end{split}
\end{equation}
where the set $\mathcal{N}_l$ contains the three agents that form the $l$-th triangle. Note that according to the Hamiltonian (\ref{H}), $c_{\theta_l n}, c_{\phi_l n} \in \mathbb{R}_+$ are respectively the sum of three elements which are related to the agents of the corresponding triangle $l$, where $l \in \{1,2,...,M\}$. $c_{\epsilon} \in \mathbb{R}_+$ is the sum of parameters related to all estimators of edge $\epsilon$ by the corresponding agent.  

It follows that $H^L$ is positive definite. Now we consider the time derivative of (\ref{gH}) using the same procedure as in Section 3.4, leading to
\begin{equation} \label{dgH}
\begin{split}
\dot{H}^L=&-\sum_{n=1}^N D^r_n\dot q_n^T\dot q_n -\sum_{l=1}^{M}\sum_{n \in \mathcal{N}_l}(d_{\theta_l n}\dot{\widetilde{\cos{\theta_l}}}^2+d_{\phi_l n}\dot{\widetilde{\cos{\phi_l}}}^2)\leq0.
\end{split}
\end{equation}
By invoking LaSalle's Invariance principle, we get that the trajectories of the closed-loop system converge to the largest invariant set where $\dot{H}^L=0$. On this set $\dot{q}_n=0$ for all $n \in \{1, 2, ... ,N\}$ and $\dot{\widetilde{\cos{\theta_l}}}=0$ and $\dot{\widetilde{\cos{\phi_l}}}=0$ for all $l \in \{1,2,...,M\}$.

Furthermore, we can conclude that $\ddot{q}_n=0$ for all $n \in \{1, 2, ... ,N\}$. According to the control law, it follows that 
\begin{equation} \label{Uneq0}
U_n=0, \quad n \in \{1, 2, ... ,N\}
\end{equation}
on this invariant set. Substituting $\dot{\widetilde{\cos{\theta_l}}}=0$ and $\dot{\widetilde{\cos{\phi_l}}}=0$ for all $l \in \{1,2,...,M\}$ into (\ref{Uneq0}), we get
\begin{equation} \label{Lsui}
\begin{split}
0= \matr{c_{\theta_l}\hat{L}_{\theta_l}^T & c_{\phi l}\hat{L}_{\phi_l}^T}\matr{\widetilde{\cos{\phi_l}} \\ \widetilde{\cos{\phi_l}}}, \quad l \in \{1,2,...,M\}.
\end{split}
\end{equation}
Similarly to the form in (\ref{sui}), $\matr{c_{\theta_l}\hat{L}_{\theta_l}^T & c_{\phi l}\hat{L}_{\phi_l}^T}$ is a $6 \times 2$ matrix. According to $\textit{Lemma 3}$, $\hat{L}_l^T$ is of full column rank for all $l \in \{1,2,...,M\}$. Note that $c_{\theta_l}, c_{\phi_l}, l \in \{1,2,...,M\}$ are positive scalars, and similarly to the proof of theorem 1, the matrix $\matr{c_{\theta_l}\hat{L}_{\theta_l}^T & c_{\phi l}\hat{L}_{\phi_l}^T}$ is also of full column rank. Therefore, we conclude that $\widetilde{\cos{\theta_l}}$ and $\widetilde{\cos{\phi_l}}$ are zero. It follows that $\theta_l=\theta_l^*$ and $\phi_l=\phi_l^*$ for all $l \in \{1,2,...,M\}$, thus completing the proof.
\end{proof}

\section{Formation maneuvers}
%
%

In order to complete a task, it is necessary to maneuver the whole formation. With this aim, we design controllers for scale and orientation control and velocity tracking.  

For scale and orientation control, without loss of generality, we choose the edge $k$ connecting reference agents $1, 2$ as the reference edge. Assume the coordinates of agents $1, 2$ are aligned and the desired displacement with pre-specified scale and orientation can be described as $z_k^*=(q_1-q_2)^*$. Then, the corresponding Hamiltonian is given as
\begin{equation} \label{Hs}
\begin{split}
H_{k}^z=\frac{1}{2}(z_k-z_k^*)^T(z_k-z_k^*).
\end{split}
\end{equation}

The dynamics of the controller are designed by assigning virtual couplings along the edge $k$. The resulting control law with spring and damping terms associated with edge $k$ is given as
\begin{equation} \label{csk}
\begin{split}
U_k^z=-\frac{\partial{H_k}}{\partial{z_k}}-D_z\dot{z}_k,
\end{split}
\end{equation}
where $D_z \in \mathbb{R}_+^{2 \times 2}$ is a constant damping matrix. Furthermore, using the incidence matrix to connect the edge and the agents, we have
\begin{equation} \label{cs}
\begin{split}
U_1^z=\frac{\partial{H_k}}{\partial{z_k}}+D_z\dot{z}_k,  \\
U_2^z=-\frac{\partial{H_k}}{\partial{z_k}}-D_z\dot{z}_k.
\end{split}
\end{equation}

\begin{remark}
The specification of any of the displacements associated with the edges is sufficient to define the scale and orientation of the network with the underlying triangulated Laman graph. Regarding the scale, considering the triangle that the edge with specified displacement forms, since three inner angles and one distance is known in this triangle, the distances of the other two edges are also uniquely determined. Moreover, the three edges also form other triangles whose inner angles are uniquely determined due to the angle-based formation stabilization. In a similar fashion, the distances of all edges are uniquely determined, i.e., the formation is angle and distance rigid. Concerning the orientation, since the formation is angle and distance rigid, one bearing is sufficient to determine the orientation of the formation.   
\end{remark} 

For velocity tracking, we assume all agents know the desired velocity $v^*$. The corresponding Hamiltonian for each agent is given as 
\begin{align*}
\begin{split}
H_{n}^{v}=\frac{1}{2m_n}(p_n-p_n^*)^T(p_n-p_n^*),
\end{split}
\end{align*} 
where $p_n^*$ is the desired momentum. Similarly, the control law with spring and damping terms can be given as
\begin{equation} \label{uv}
\begin{split}
U_{n}^v=-D_n^rv^*-D_n^v(v_n-v^*),
\end{split}
\end{equation} 
where $D_n^v \in \mathbb{R}_+^{2 \times 2} > 0$ is the damping constant matrix.

\begin{remark}
By using velocities to estimate distances, our framework permits to inject damping into the formation maneuvers. As shown in (\ref{csk}), the term $D_z\dot{z}_k$ improves the transient performance on converging to desired scale and orientation. As for (\ref{uv}), the term $D_n^rv^*$ is used to determine the equilibrium of velocity, and the term $D_n^v(v_n-v^*)$ is used to inject damping for velocity tracking. Although assuming there is friction, $D_n^rv^*$ combined with $D_n^rv_n$ induced by friction can be seen as damping, $D_n^r$ is determined by friction and not tuneable.    
\end{remark}

To sum up, the result of formation maneuvers is given by the following theorem.
\begin{theorem} \label{thm3}
Consider a group of agents modeled as in (\ref{model}) with an underlying triangulated Laman graphs $\mathcal{G}_n(\mathcal{V}_n,\mathcal{E})$ with $M$ triangles. The angles to be controlled are defined by choosing any two angles of each triangle whose values are away enough from $0$ or $\pi$. The desired angles are achieved by the control law (\ref{um}) proposed in Theorem 2, while the desired scale and orientation are achieved by the control law (\ref{cs}). The velocity tracking is achieved by the control law (\ref{uv}).
\end{theorem}
 
Since the dynamics of angle-based formation stabilization and maneuvers are not coupled, the proof of scale and orientation control, and trajectory tracking is straightforward. We omit the details here and leave them to the readers. 

\section{Simulations}
Consider a group of four agents modeled as in (\ref{model}). The network diagram of the formation with heterogeneous constraints is shown in Fig. 3. The initial positions and control objectives are given in Table \ref{tb:conditions} and the related parameters are given in Table \ref{tb:parameters}.

\begin{figure} [!ht] 
\begin{center}
\label{Lamangraph}
\includegraphics[width=4.0cm]{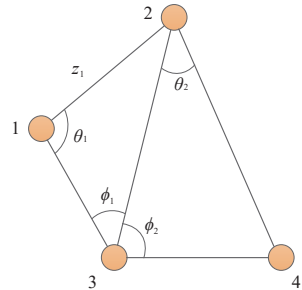}   
\caption{Formation diagram with 2 triangles} 
\end{center}
\end{figure}

\begin{table} [!ht]
\begin{center}
\caption{Initial positions and desired angles}\label{tb:conditions}
\begin{tabular}{cccc}
Parameter & Value & Parameter & Value \\\hline
$q_1(0)$ & (0.9,2.2) &$q_2(0)$ & (2.2,2.8)   \\
$q_3(0)$ & (1.8,1.1) & $q_4(0)$ & (3.1,1.9)  \\
$v_{1,2,3,4}(0)$ & (0,0) & $v^*_{1,2,3,4}$ & (3,3)  \\
$\theta_1^*$ & $\pi/2$ & $\phi_2^*$ & $\pi/4$  \\
$\theta_2^*$ & $\pi/4$ & $\phi_2^*$ & $\pi/4$ \\
$s_1^*$ & $(10,-10)$ \\\hline
\end{tabular}
\end{center}
\end{table}

\begin{table} [!ht]
\begin{center}
\caption{Model and controller parameters}\label{tb:parameters}
\begin{tabular}{cccc}
Parameter & Value \\\hline
$m_n, n\in \{1,2,3,4\}$ & 1   \\
$ c_{\theta_l n}, l\in \{1,2\}, n\in \{1,2,3,4\}$ & 10  \\
$ c_{\phi_l n}, l\in \{1,2\}, n\in \{1,2,3,4\}$ & 10  \\
$ d_{\theta_l n}, l\in \{1,2\}, n\in \{1,2,3,4\}$ & 1  \\
$ d_{\phi_l n},l\in \{1,2\}, n\in \{1,2,3,4\}$ & 1  \\
$ c_\epsilon, \epsilon \in \mathcal{E}$ & 10  \\
$ D_z$ & ${\rm{diag}}$(0.8,0.8)  \\
$ D_n^v, n\in \{1,2,3,4\}$ & ${\rm{diag}}$(1.8,1.8)
\\\hline
\end{tabular}
\end{center}
\end{table}

By implementing the controller given by (\ref{um}), (\ref{cs}), and (\ref{uv}), with the corresponding estimators, we obtain the results depicted in Figs. 4--6. In particular, we observe that the errors converge to zero in Figs. 4 and 5, illustrating the result of Theorem \ref{thm3}. The trajectory evolution of the four agents is shown in Fig. 6, where the initial positions and terminal positions (TP) are enlarged.

\begin{figure}[!ht]    
    \centering
    \subfigure[Error evolution of $\theta_1$]{\includegraphics[width=4cm]{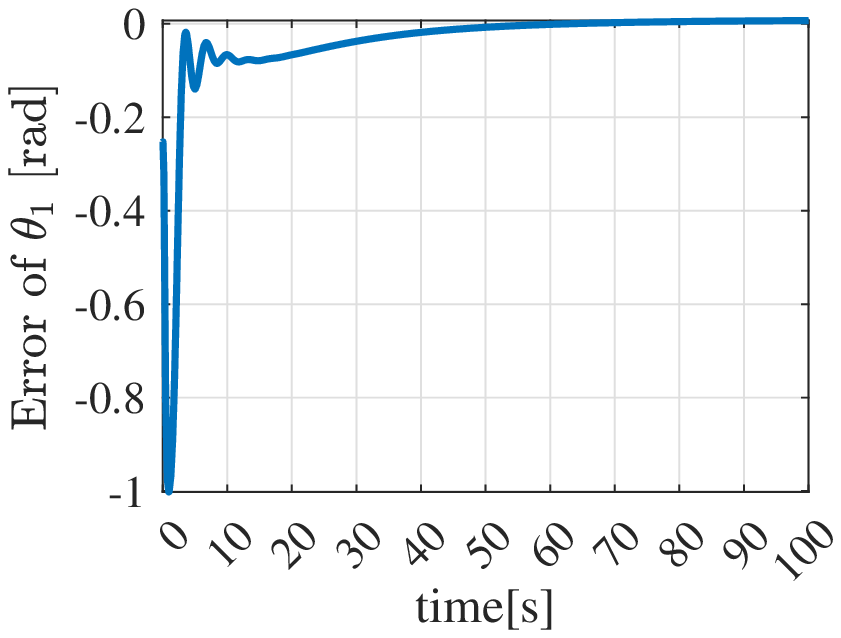}}
    \subfigure[Error evolution of $\phi_1$]{\includegraphics[width=4cm]{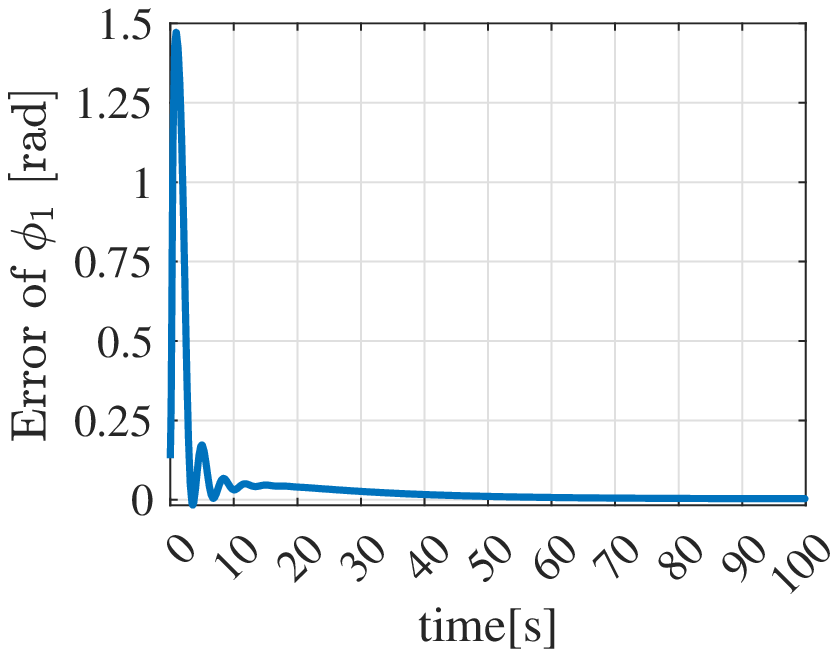}}
    \quad 
    \subfigure[Error evolution of $\theta_2$]{\includegraphics[width=4cm]{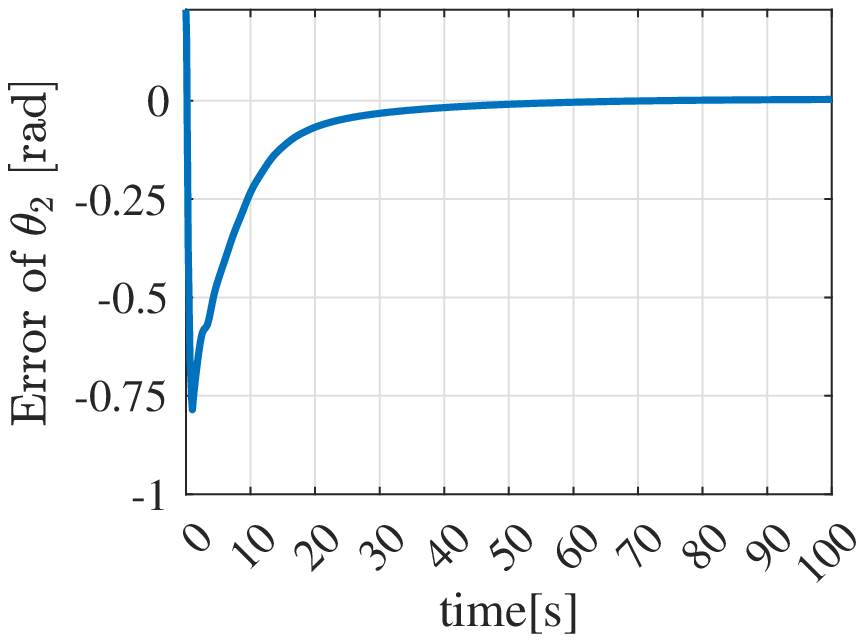}}
    \subfigure[Error evolution of $\phi_2$]{\includegraphics[width=4cm]{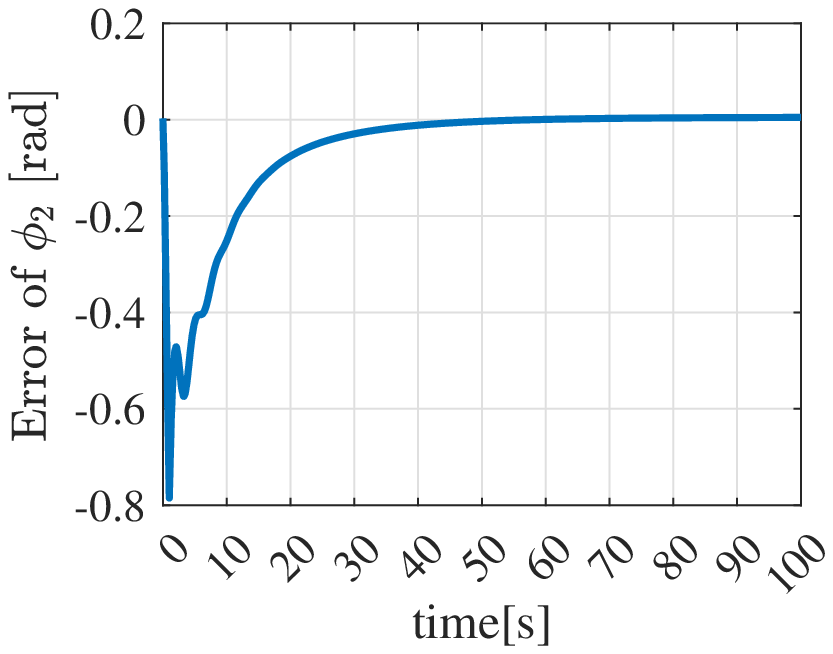}}
    \quad 
    \caption{Error evolution of the formation stabilization}
\end{figure}

\begin{figure}[!ht]    
    \centering
    \subfigure[Error evolution of the displacement]{\includegraphics[width=4cm]{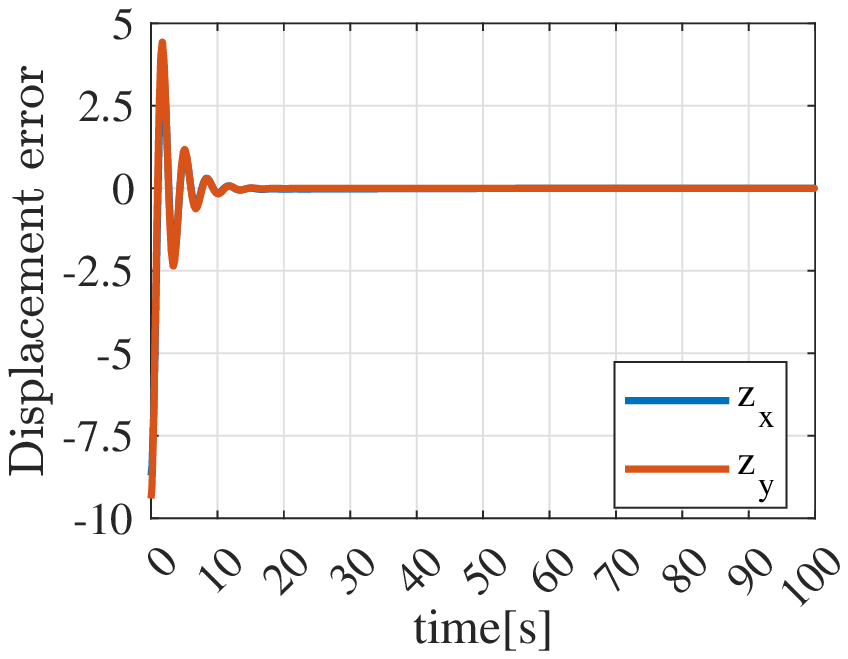}}
    \subfigure[Error evolution of velocities tracking]{\includegraphics[width=4cm]{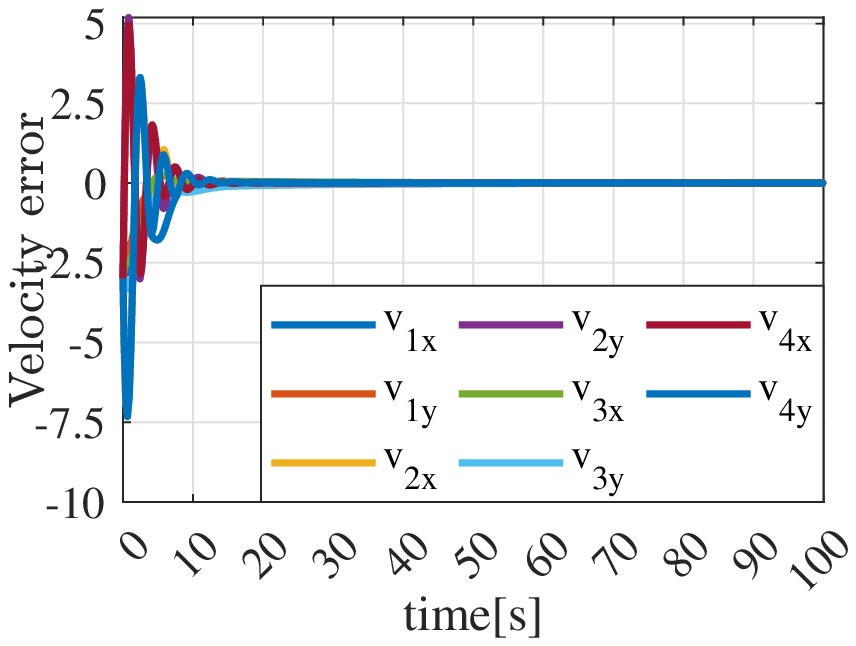}}
    \quad 
    \caption{Error evolution of maneuvers}
\end{figure}

\begin{figure}[!ht]    
    \centering
    \includegraphics[width=6.0cm]{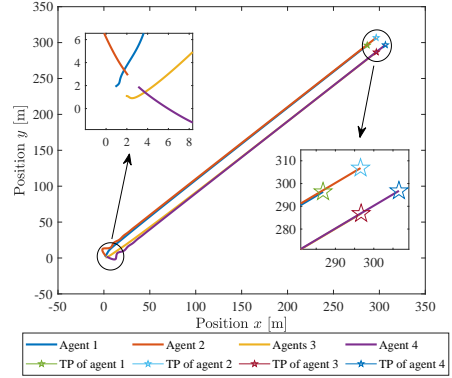} 
    \caption{Formation trajectory}
\end{figure}

\section{Conclusions and future research}
In this paper, we have developed a port-Hamiltonian framework for a network of agents modeled as double integrators to achieve a formation with an underlying triangulated Laman graph constrained by the desired angles. The formation controller is derived by assigning the spring and damping couplings along the angles to be controlled and the mappings between the angles and agent actuators. An estimator is proposed by pH theory and the fact that energy is coordinate-free for the unavailable distances. In addition, several formation maneuvers are investigated, and damping injection is also achieved using velocity measurement.  

Regarding future research, our framework can be extended to multi-agent formations with heterogeneous dynamics and formation constraints. On the one hand, the pH modeling allows for complex and heterogeneous agent dynamics and the Dirac structure behind it enables the scalability of the network. On the other hand, the principle of the controller is to assign the virtual couplings on the error of formation constraints and then to map these couplings to the actuators by the measurement Jacobians, which provides a general framework for other formation constraints.

\bibliographystyle{unsrt}
\bibliography{ifacconf}             

\begin{thebibliography}{10}

\bibitem{beard2001coordination}
Randal~W. Beard, Jonathan Lawton, and Fred~Y. Hadaegh.
\newblock A coordination architecture for spacecraft formation control.
\newblock {\em IEEE Transactions on control systems technology}, 9(6):777--790,
  2001.

\bibitem{ren2007information}
Wei Ren, Randal~W. Beard, and Ella~M. Atkins.
\newblock Information consensus in multivehicle cooperative control.
\newblock {\em IEEE Control systems magazine}, 27(2):71--82, 2007.

\bibitem{oh2015survey}
Kwang~Kyo Oh, Myoung~Chul Park, and Hyo~Sung Ahn.
\newblock A survey of multi-agent formation control.
\newblock {\em Automatica}, 53:424--440, 2015.

\bibitem{cortes2017coordinated}
Jorge Cort{\'e}s and Magnus Egerstedt.
\newblock Coordinated control of multi-robot systems: A survey.
\newblock {\em SICE Journal of Control, Measurement, and System Integration},
  10(6):495--503, 2017.

\bibitem{liu2018survey}
Yuanchang Liu and Richard Bucknall.
\newblock A survey of formation control and motion planning of multiple
  unmanned vehicles.
\newblock {\em Robotica}, 36(7):1019--1047, 2018.

\bibitem{chen2019control}
Fei Chen and Wei Ren.
\newblock On the control of multi-agent systems: A survey.
\newblock {\em Foundations and Trends{\textregistered} in Systems and Control},
  6(4):339--499, 2019.

\bibitem{nuno2020distributed}
Emmanuel Nuno, Antonio Loria, Tonatiuh Hernandez, Mohamed Maghenem, and Elena
  Panteley.
\newblock Distributed consensus-formation of force-controlled nonholonomic
  robots with time-varying delays.
\newblock {\em Automatica}, 120:109114, 2020.

\bibitem{ji2007distributed}
Meng Ji and Magnus Egerstedt.
\newblock Distributed coordination control of multiagent systems while
  preserving connectedness.
\newblock {\em IEEE Transactions on Robotics}, 23(4):693--703, 2007.

\bibitem{anderson2008rigid}
Brian D.~O. Anderson, Changbin Yu, Baris Fidan, and Julien~M. Hendrickx.
\newblock Rigid graph control architectures for autonomous formations.
\newblock {\em IEEE Control Systems Magazine}, 28(6):48--63, 2008.

\bibitem{cao2011formation}
Ming Cao, Changbin Yu, and Brian D.~O. Anderson.
\newblock Formation control using range-only measurements.
\newblock {\em Automatica}, 47(4):776--781, 2011.

\bibitem{zhao2019bearing}
Shiyu Zhao, Zhenhong Li, and Zhengtao Ding.
\newblock Bearing-only formation tracking control of multiagent systems.
\newblock {\em IEEE Transactions on Automatic Control}, 64(11):4541--4554,
  2019.

\bibitem{trinh2018bearing}
Minh~Hoang Trinh, Shiyu Zhao, Zhiyong Sun, Daniel Zelazo, Brian D.~O. Anderson,
  and Hyo~Sung Ahn.
\newblock Bearing-based formation control of a group of agents with
  leader-first follower structure.
\newblock {\em IEEE Transactions on Automatic Control}, 64(2):598--613, 2018.

\bibitem{jing2019angle}
Gangshan Jing, Guofeng Zhang, Heung Wing~Joseph Lee, and Long Wang.
\newblock Angle-based shape determination theory of planar graphs with
  application to formation stabilization.
\newblock {\em Automatica}, 105:117--129, 2019.

\bibitem{chen2020angle}
Liangming Chen, Ming Cao, and Chuanjiang Li.
\newblock Angle rigidity and its usage to stabilize multi-agent formations in
  2d.
\newblock {\em IEEE Transactions on Automatic Control}, 2020.

\bibitem{li2021angle}
Ningbo Li, Pablo Borja, Arjan van~der Schaft, Jacquelien~MA Scherpen, and
  Liangming Chen.
\newblock Angle formation of double integrator with bearing and velocity
  information.
\newblock {\em IFAC-PapersOnLine}, 54(19):217--222, 2021.

\bibitem{vos2014formation}
Ewoud Vos, Jacquelien Scherpen, Arjan van~der Schaft, and Ate Postma.
\newblock Formation control of wheeled robots in the port-hamiltonian
  framework.
\newblock {\em IFAC Proceedings Volumes}, 47(3):6662--6667, 2014.

\bibitem{stacey2015passivity}
Geoff Stacey and Robert Mahony.
\newblock A passivity-based approach to formation control using partial
  measurements of relative position.
\newblock {\em IEEE Transactions on Automatic Control}, 61(2):538--543, 2015.

\bibitem{xu2018formation}
Ming Xu and Yuying Liang.
\newblock Formation flying on elliptic orbits by hamiltonian
  structure-preserving control.
\newblock {\em Journal of Guidance, Control, and Dynamics}, 41(1):294--300,
  2018.

\bibitem{li2022passivity}
Ningbo Li, Jacquelien Scherpen, Arjan Van~der Schaft, and Zhiyong Sun.
\newblock A passivity approach in port-hamiltonian form for formation control
  and velocity tracking.
\newblock In {\em 2022 European Control Conference (ECC)}, pages 1844--1849.
  IEEE, 2022.

\bibitem{van2000l2}
Arjan~J Van~der Schaft.
\newblock {\em L2-gain and passivity techniques in nonlinear control},
  volume~2.
\newblock Springer, 2000.

\bibitem{khalil2002nonlinear}
Hassan~K Khalil.
\newblock {\em Nonlinear Systems}.
\newblock Third Edition, Prentice Hall, New Jersy, 2002.

\bibitem{van2013port}
Arjan~J van~der Schaft and B.~M. Maschke.
\newblock Port-{H}amiltonian systems on graphs.
\newblock {\em SIAM Journal on Control and Optimization}, 51(2):906--937, 2013.

\bibitem{mahony2012port}
Robert Mahony and Stefano Stramigioli.
\newblock A port-hamiltonian approach to image-based visual servo control for
  dynamic systems.
\newblock {\em The International Journal of Robotics Research},
  31(11):1303--1319, 2012.

\bibitem{basiri2010distributed}
Meysam Basiri, Adrian~N. Bishop, and Patric Jensfelt.
\newblock Distributed control of triangular formations with angle-only
  constraints.
\newblock {\em Systems \& Control Letters}, 59(2):147--154, 2010.

\bibitem{chen2022globally}
Liangming Chen and Zhiyong Sun.
\newblock Globally stabilizing triangularly angle rigid formations.
\newblock {\em IEEE Transactions on Automatic Control}, 2022.

\bibitem{chen2021maneuvering}
Liangming Chen, Hector~Garcia de~Marina, and Ming Cao.
\newblock Maneuvering formations of mobile agents using designed mismatched
  angles.
\newblock {\em IEEE Transactions on Automatic Control}, 67(4):1655--1668, 2021.

\bibitem{van2014port}
Arjan~J van~der Schaft and Dimitri Jeltsema.
\newblock Port-{H}amiltonian systems theory: An introductory overview.
\newblock {\em Foundations and Trends{\textregistered} in Systems and Control},
  1(2-3):173--378, 2014.

\bibitem{duindam2009modeling}
Vincent Duindam, Alessandro Macchelli, Stefano Stramigioli, and Herman
  Bruyninckx.
\newblock {\em Modeling and control of complex physical systems: the
  port-Hamiltonian approach}.
\newblock Springer Science \& Business Media, 2009.

\bibitem{chen2017global}
Xudong Chen, Mohamed~Ali Belabbas, and Tamer Basar.
\newblock Global stabilization of triangulated formations.
\newblock {\em SIAM Journal on Control and Optimization}, 55(1):172--199, 2017.

\end{thebibliography}

\end{document}